\documentclass[11pt]{article}

\usepackage{fullpage}

\usepackage{times}
\usepackage{soul}
\usepackage{url}
\usepackage[utf8]{inputenc}
\usepackage[small]{caption}
\usepackage{graphicx}
\usepackage{amsmath}
\usepackage{booktabs}
\usepackage{enumitem}
\usepackage[linesnumbered,boxed,vlined,ruled]{algorithm2e}

\SetArgSty{textrm}  

\SetCommentSty{mycommfont}

\usepackage{libertine}

\usepackage{amssymb,amsmath,amsfonts,amstext,amsthm}

\usepackage{hyperref}
\usepackage[svgnames]{xcolor}
\usepackage[capitalise,nameinlink]{cleveref}
\hypersetup{colorlinks={true},linkcolor={DarkBlue},citecolor=[named]{DarkGreen}}

\usepackage{natbib}
\usepackage{multirow}

\newtheorem{theorem}{Theorem}

\newtheorem{claim}[theorem]{Claim}

\newtheorem{cor}[theorem]{Corollary}

\theoremstyle{definition}

\newcommand{\calI}{\mathcal{I}}
\newcommand{\calC}{\mathcal{C}}
\newcommand{\SW}{\text{SW}}
\newcommand{\vv}{\mathbf{v}}
\newcommand{\tol}[1]{\mathbf{t_{#1}}}

\newcommand{\EQ}{\text{EQ}}
\newcommand{\OPT}{\text{OPT}}
\newcommand{\PoA}{\text{PoA}}
\newcommand{\PoS}{\text{PoS}}

\allowdisplaybreaks

\begin{document}

\title{\bf Not all Strangers are the Same: \\ The Impact of Tolerance in Schelling Games}

\author{Panagiotis Kanellopoulos \and Maria Kyropoulou \and Alexandros A. Voudouris}

\date{School of Computer Science and Electronic Engineering \\ University of Essex, UK}

\maketitle

\begin{abstract}
Schelling’s famous model of segregation assumes agents of different types, who would like to be located in neighborhoods having at least a certain fraction of agents of the same type. We consider natural generalizations that allow for the possibility of agents being tolerant towards other agents, even if they are not of the same type. In particular, we consider an ordering of the types, and make the realistic assumption that the agents are in principle more tolerant towards agents of types that are closer to their own according to the ordering. Based on this, we study the strategic games induced when the agents aim to maximize their utility, for a variety of tolerance levels. We provide a collection of results about the existence of equilibria, and their quality in terms of social welfare. 
\end{abstract}

\section{Introduction} \label{sec:intro}
Residential segregation is a broad phenomenon affecting most metropolitan areas, and is known to be caused due to racial or socio-economic differences.  The severity of its implications to society~\citep{BMR01} is the main reason for the vast research attention it has received, with many different models being proposed over the years that aim to conceptualize it (e.g., see \cite{T56}).  The most prominent of those models is that of \citet{schelling1969,schelling1971}, which studies how motives at an individual level can lead to macroscopic behavior and, ultimately, to segregation. In particular, the individuals are modelled as agents of two different types (usually referred to using colors, such as red and blue), and the environment is abstracted by a topology (such as a grid graph), representing a city. The agents occupy nodes of the topology, and prefer neighborhoods in which the presence of their own type exceeds a specified tolerance threshold. If an agent is unhappy with her current location, then she either jumps to a randomly selected empty node of the topology, or swaps positions with another random unhappy agent. Schelling's crucial observation was that such dynamics might lead to largely segregated placements, even when the agents are relatively tolerant of mixed neighborhoods.

A recent series of papers (discussed in Section~\ref{sec:related}) have generalized Schelling's model to include more than two types, and have taken a game-theoretic approach, according to which the agents behave strategically rather than randomly, aiming to maximize their individual {\em utility}. There are many ways to define the utility of an agent $i$ of type $T$. For instance, \citet{elkind2019jump} defined it as the ratio of the number of agents of type $T$ in $i$'s neighborhood over the {\em total} number of agents therein. \citet{echzell2019dynamics} proposed a similar definition, which however does not take into account all the agents of different type in the denominator, but only those of the {\em majority} type. The first definition essentially assumes that the agents view all the agents of different type as {\em enemies}. On the other hand, the second definition assumes that the agents view only the majority type as hostile. An alternative way of thinking about these particular utility functions is as if the agents have binary {\em tolerance} towards other agents in the sense that agents are either friends or enemies; in the case of \citeauthor{elkind2019jump} all the neighbors of an agent are taken into account when computing her utility, whereas in the case of \citeauthor{echzell2019dynamics} some of her neighbors are ignored.

These functions are natural generalizations of the quantity that determines the happiness of agents in Schelling's original model for two types. However, they fail to capture realistic scenarios in which the agents do not have a single-dimensional view of the other agents, but rather have different {\em preferences} over the different types of agents. As an example, consider people looking to buy houses, who may be of four possible types: young single, in a couple, in a family with children, or retired. Such individuals might have preferences over the types, and these preferences are likely to be consistent with the specific given ordering of the types. That is, a young single person would rather live around other single people, or even young couples as they are more likely to establish social relationships, less so with families and even less so with elderly. Similarly, a family with children would be happy to be around other families, but would not mind (young or retired) couples as they are likely to prefer a safe and quiet neighborhood, where their children can grow up and make friends.

\subsection{Our contribution}
To capture scenarios like the example above, we propose a clean model that naturally extends Schelling's original model by incorporating {\em different levels of tolerance} among agent types, and study the induced strategic games in terms of the existence and quality of their equilibria; in Section~\ref{sec:conclusions}, we discuss potential generalizations of our model. 


To be more specific, our model consists of a set of agents who are partitioned into $\lambda \geq 2$ types, a graph topology, and an ordering of the different agent types which determines the relative tolerance among agents of different types. Naturally, we assume that there is higher tolerance between agents whose types are closer according to the ordering. The exact degree of tolerance between the different types is specified by a {\em tolerance vector}, which consists of weights representing the tolerance between the different types depending on their distance in the given ordering. For example, agents of the same type are in distance $0$ and are fully tolerant towards each other, which is captured by a weight of $1$. The utility of an agent can then be computed as a weighted average of the tolerance that she has towards her neighbors, and every agent aims to occupy a node of the topology to maximize her utility; agents are allowed to unilaterally {\em jump} to empty nodes to increase their utility. 

We study the dynamics of such {\em tolerance Schelling games}. We first focus on questions related to equilibrium existence. For general games, we show that equilibria are not guaranteed to exist if agents are not fully tolerant towards agents in type-distance $1$. We complement this impossibility by showing many positive results for important subclasses of games, in which the topology is a structured graph (such as a grid or a tree) and the tolerance vector satisfies certain properties. We then turn our attention to the quality of equilibria measured by the social welfare objective, defined as the total utility of the agents, and prove (asymptotically tight) bounds on the price of anarchy~\citep{poa} and price of stability~\citep{pos}. 

\subsection{Related work} \label{sec:related}
Residential segregation, and Schelling's original randomized model in particular, has been the basis of a continuous stream of multidisciplinary research in 
Sociology \citep{Clark2008understanding}, 
Economics \citep{Pancs2007spatial,Zhang2004residential}, 
Physics \citep{Vinkovic2006schelling}, 
and Computer Science \citep{Barmpalias2014digital,Bhakta2014clustering,Blasius2021flip,Brandt2012one,Immorlica2017exponential}. 

Most related to our work is a quite recent series of papers in the TCS and AI communities, which deviated from the premise of random behavior, and instead studied the strategic games induced when the agents act as utility-maximizers. \citet{chauhan2018schelling} studied questions related to dynamics convergence in games with two types of agents who can either {\em jump} to empty nodes of a topology (as in our case) or {\em swap} locations with other agents to minimize a {\em cost} function; their model was generalized to multiple types of agents by \citet{echzell2019dynamics}. In this paper we extend the {\em utility} model of \citet{elkind2019jump}, who initially refined the model of \citet{chauhan2018schelling}. They introduced a simpler utility function which the agents aim to maximize, and studied the existence, complexity and quality of equilibria in jump games with multiple types of agents and general topologies. They also proposed many interesting variants, such as {\em enemy aversion} and {\em social Schelling games}, which have been partially studied by \citet{KKV-modified} and \citet{chan2020schelling}, respectively.  \citet{Agarwal2020swap} studied similar questions for swap games, and   
\citet{bilo2020topological} considered a constrained setting, in which the agents can only view a small part of the topology near their current location. Finally, \citet{Bullinger2021pareto} studied the complexity of computing assignments with good welfare guarantees, such as Pareto optimality and variants of it.


\section{Preliminaries}\label{sec:prelims}
A {\em $\lambda$-type tolerance Schelling game} consists of:
\begin{itemize}
\item A set $N$ of  $n \geq 4$ {\em agents}, partitioned into $\lambda \geq 2$ disjoint sets $T_1, \ldots, T_\lambda$ representing {\em types}, such that $\bigcup_{\ell \in [\lambda]} T_\ell = N$.
\item A simple connected undirected graph $G=(V,E)$ called {\em topology}, such that $|V| > n$.
\item A {\em tolerance vector} $\tol{\lambda} = [t_0, \ldots, t_{\lambda-1}]$ consisting of $\lambda$ parameters, such that $t_d$ represents the tolerance that agents of type $T_\ell$ have towards agents of type $T_k$ in Manhattan distance $|\ell-k| = d \in \{0, \ldots, \lambda-1\}$ according to a given ordering $\succ$ of the types (say, $T_1 \succ \ldots \succ T_\lambda$). We assume that agents are more tolerant towards agents of types that are closer to their own according to $\succ$, and we thus have that $1 = t_0  \geq \ldots \geq t_{\lambda-1} \geq 0$. We also assume that $t_{\lambda-1} < 1$; otherwise, all agents are completely tolerant towards all others and the game is trivial. Let $\tau = \sum_{d = 0}^{\lambda-1} t_d$ be the sum of all tolerance parameters.
\end{itemize}
Clearly, the class of $\lambda$-type tolerance Schelling games includes as a special case the classic Schelling games studied in the related literature (e.g., see \citep{elkind2019jump}), for which $t_0 = 1$ and $t_d = 0$ for every $d \in \{1, \ldots, \lambda-1\}$. Because of this particular tolerance vector, we will use the term {\em $\lambda$-type zero-tolerance games} to refer to the classic Schelling games.

In this paper we consider \emph{balanced} games, in which the agents are partitioned in types of equal size, such that $|T_\ell| = n/\lambda \geq 2$ for every $\ell \in [\lambda]$; thus, $n$ is a multiple of $\lambda$. This assumption makes the analysis more interesting as it helps bypass straightforward negative results. We use the abbreviation {\em $\lambda$-TS} to refer to such a balanced $\lambda$-type tolerance Schelling game game $\calI = (N,G,\tol{\lambda})$. For convenience, we will also use the abbreviation {\em $\lambda$-ZTS} to refer to a balanced $\lambda$-type zero-tolerance game $\calI = (N,G)$.

Let $\vv = (v_i)_{i \in N}$ be an {\em assignment} specifying the node $v_i$ of $G$ that each agent $i \in N$ {\em occupies}, such that $v_i \neq v_j$ for $i \neq j$. For every node $v$, we denote by $n_\ell(v|\vv)$ the number of agents of type $T_\ell$ that occupy nodes in the {\em neighborhood} of $v$ according to the assignment $\vv$, and also let $n(v |\vv)=\sum_{\ell \in [\lambda]}n_\ell(v | \vv)$. Given an assignment $\vv$, the utility of agent $i$ of type $T_\ell$ is computed as
\begin{align*}
u_i(\vv) = \frac{1}{n(v_i|\vv)} \sum_{k \in [\lambda]} t_{|\ell-k|} \cdot n_k(v_i|\vv),
\end{align*}
if $n(v_i|\vv)\neq 0$, and $0$ otherwise (in which case we say that the agent is \emph{isolated}).
The agents are {\em strategic} and aim to maximize their utility by {\em jumping} to empty nodes of the topology if they can increase their utility by doing so. We say that an assignment $\vv$ is an {\em equilibrium} if no agent $i$ of any type $T_\ell$ has incentive to jump to any empty node $v$ of the topology, that is, $u_i(\vv) \geq u_i(v,\vv_{-i})$, where $(v,\vv_{-i})$ is the assignment resulting from this jump. Let $\EQ(\calI)$ be the set of equilibrium assignments of a given $\lambda$-TS game $\calI$.

The {\em social welfare} of an assignment $\vv$ is defined as the total utility of the agents, that is,
$$\SW(\vv) = \sum_{i \in N} u_i(\vv).$$
Let $\OPT(\calI) = \max_{\vv} \SW(\vv)$ be the maximum social welfare among all possible assignments in the $\lambda$-TS game $\calI$. For a given subclass $\calC$ of $\lambda$-TS games, the {\em price of anarchy} is defined as the worst-case ratio, over all possible games $\calI \in \calC$ such that $\EQ(\calI) \neq \varnothing$, between $\OPT(\calI)$ and the {\em minimum} social welfare among all equilibria:
$$\PoA(\calC) = \sup_{\calI \in \calC: \EQ(\calI) \neq \varnothing} \frac{\OPT(\calI)}{\min_{\vv \in \EQ(\calI)}\SW(\vv)}.$$
Similarly, the {\em price of stability} takes into account the ratio between $\OPT(\calI)$ and the {\em maximum} social welfare among all equilibria:
$$\PoS(\calC) = \sup_{\calI \in \calC: \EQ(\calI) \neq \varnothing} \frac{\OPT(\calI)}{\max_{\vv \in \EQ(\calI)}\SW(\vv)}.$$


\section{Equilibrium existence}\label{sec:existence}
In this section, we show several positive and negative results about the existence of equilibrium assignments, for interesting subclasses of tolerance Schelling games. We start with the relation of equilibrium assignments in $\lambda$-ZTS games and general $\lambda$-TS games. 

\begin{theorem} \label{thm:2-type-equivalence}
Consider a $\lambda$-ZTS game $\calI=(N,G)$ and a $\lambda$-TS game $\calI'=(N,G,\tol{\lambda})$.
For $\lambda=2$, $\EQ(\calI') \subseteq \EQ(\calI)$ and $\EQ(\calI) \setminus \EQ(\calI')$ consists of assignments with isolated agents. 
For $\lambda \geq 3$, $\EQ(\calI)$ and $\EQ(\calI')$ are incomparable.
\end{theorem}

\begin{proof}
We start with $\lambda=2$; for convenience, we will refer to the two types as red and blue. 
Let $\vv$ be an equilibrium of $\calI'$. Clearly, for $\calI$ and $\calI'$ to be different, it must be the case that $t_1 > 0$. Consequently, there are no isolated agents in $\vv$ as they would have incentive to deviate to nodes that are adjacent to agents of the other type and increase their utility from $0$ to $t_1$. We will show that $\vv$ is an equilibrium of $\calI$ as well. Without loss of generality, consider a red agent $i$ who occupies a node $v_i$ that is adjacent to $n_r(v_i)$ red and $n_b(v_i)$ blue agents. Since agent $i$ is not isolated, it holds that $n_r(v_i) + n_b(v_i) \geq 1$. If $n_b(v_i) = 0$, then agent $i$ has maximum utility $1$ in both $\calI$ and $\calI'$. 
Hence, we can assume that $n_b(v_i) \geq 1$. 
Since $\vv$ is an equilibrium of $\calI'$, agent $i$ has no incentive to unilaterally jump to any empty node $v$ of the topology. That is,
\begin{align*}
&\frac{n_r(v_i) + t_1 \cdot n_b(v_i) }{n_r(v_i) + n_b(v_i)} \geq \frac{n_r(v) + t_1 \cdot n_b(v)}{n_r(v) + n_b(v)} \\
&\Leftrightarrow  (1-t_1) \left( \frac{n_r(v_i)}{n_b(v_i)} - \frac{n_r(v)}{n_b(v)} \right) \geq 0,
\end{align*}
where $n_r(v)$ and $n_b(v)$ are the number of red and blue agents that are adjacent to $v$ {\em after} agent $i$ jumps to $v$; observe that $n_b(v) \geq 1$, as otherwise agent $i$ would obtain maximum utility of $1$ by jumping to $v$, contradicting that $\vv$ is an equilibrium of $\calI'$. Since $t_1 < 1$, we equivalently have that
$$\frac{n_r(v_i)}{n_b(v_i)} \geq \frac{n_r(v)}{n_b(v)} 
\Leftrightarrow \frac{ n_r(v_i) }{ n_r(v_i) + n_b(v_i) } \geq \frac{ n_r(v) }{ n_r(v) + n_b(v)}.$$
Therefore, agent $i$ has no incentive to jump to the empty node $v$ in $\calI$, and $\vv$ is an equilibrium of $\calI$ as well. Using similar arguments, we can show that any equilibrium of $\calI$ such that there is no isolated agent is also an equilibrium of $\calI'$. 

For $\lambda \geq 3$, to show that $\EQ(\calI)$ is incomparable to $\EQ(\calI')$, consider the tolerance vector $\tol{3} = (1, 1/2, 0)$ and the following two partial assignments $\vv$ and $\vv'$:
\begin{itemize}
\item In $\vv$, an agent $i$ of type $T_1$ occupies a node $v_i$ which is adjacent to two nodes, one occupied by an agent of type $T_1$ and one occupied by an agent of type $T_3$. There is also an empty node $v$ which is adjacent to two nodes, one occupied by an agent of type $T_1$ and one occupied by an agent of type $T_2$. In $\calI$, agent $i$ has no incentive to jump from $v_i$ to $v$ as both nodes give her utility $1/2$. On the other hand, in $\calI'$, agent $i$ has utility $(1+t_2)/2 = 1/2$ and has incentive to jump to $v$ to increase her utility to $(1+t_1)/2 = 3/4$. Hence, $\vv$ can be an equilibrium of $\calI$, but not of $\calI'$. 

\item In $\vv'$, an agent $i$ of type $T_1$ occupies a node $v_i$ which is adjacent to three nodes, one occupied by an agent of type $T_1$, one occupied by an agent of type $T_2$ and one occupied by an agent of type $T_3$. There is also an empty node $v$ which is adjacent to two nodes, one occupied by an agent of type $T_1$ and one occupied by an agent of type $T_3$. In $\calI$, agent $i$ has incentive to jump from $v_i$ to $v$ in order to increase her utility from $1/3$ to $1/2$. However, in $\calI'$, agent $i$ has no incentive to jump as she has utility $(1+t_1)/3 = 1/2$ by occupying node $v_i$, which is exactly the utility she would also obtain by jumping to $v$. Consequently, $\vv'$ can be an equilibrium of $\calI'$, but not of $\calI$. 
\end{itemize}
This completes the proof.
\end{proof}

Since there exist simple $2$-ZTS games that do not admit any equilibria~\citep{elkind2019jump}, the first part of Theorem~\ref{thm:2-type-equivalence} implies that equilibria are not guaranteed to exist for general $2$-TS games as well. In fact, by carefully inspecting the proof of \citet{elkind2019jump} that $\lambda$-ZTS games played on trees do not always admit equilibria for every $\lambda \geq 2$, we can show the following stronger impossibility result.

\begin{theorem} \label{thm:inexistence}
For every $\lambda \geq 2$ and every tolerance vector $\tol{\lambda}$ such that $t_1 <1$, there exists a $\lambda$-TS game $\calI = (N,G,\tol{\lambda})$ in which $G$ is a tree and does not admit any equilibrium.
\end{theorem}

\begin{proof}
We will show how the proof of \citet{elkind2019jump} that $\lambda$-ZTS games do not always admit equilibria can be directly extended to capture the class of $\lambda$-TS games with appropriate tolerance vectors.
Consider a $\lambda$-TS game $\calI = (N,G,\tol{\lambda})$ where $N$ is a set of $n=\lambda(2\lambda+1)$ agents partitioned into $\lambda \geq 2$ types such that there are $(2\lambda+1)$ agents per type, and $\tol{\lambda}$ is a tolerance vector with $t_1 < 1$. The topology $G$ is a tree consisting of $n+1$ nodes, distributed to four levels. There is a root node $\alpha$ with one child node $\beta$, which in turn has a set $\Gamma$ of $2\lambda-1$ children. Every node $\gamma \in \Gamma$ has a set $\Delta_\gamma$ of $\lambda$ children nodes, which are leaves of the tree; let $\Delta = \bigcup_{\gamma \in \Gamma} \Delta_\gamma$. Adapting the arguments of \citet{elkind2019jump}, we will now show that this game does not admit any equilibrium. Assume otherwise, that there exists an equilibrium assignment $\vv$. We distinguish between cases by enumerating the single node $v$ that is left empty:
\begin{itemize}
\item $v = \alpha$. Suppose that the agent $i$ who occupies node $\beta$ is of type $T$.
Since $|\Gamma| = 2\lambda-1$ and $|T\setminus \{i\}| = 2\lambda$, there must be an agent $j \neq i$ of type $T$ with a neighbor of type $T' \neq T$. Hence, agent $j$ has utility strictly less than $1$, and has incentive to deviate to $\alpha$ to obtain maximum utility $1$.

\item $v = \beta$. Suppose that the agent $i$ who occupies node $\alpha$ is of type $T_\ell$. Since she is isolated, her utility is $0$. Since $\vv$ is assumed to be an equilibrium, agent $i$ would still have utility $0$ by deviating to $\beta$; note that this is impossible if all tolerance parameters are positive. As a result, all the nodes in $\Gamma$ must be occupied by agents of types $T_k$ such that $t_{|\ell-k|} = 0$. But then, this means that all agents of type $T_\ell$, who must occupy nodes of $\Delta$, also have utility $0$ and incentive to deviate to $\beta$, connect to agent $i$, and obtain positive utility.

\item $v \in \Gamma$. Since $v$ is empty, all the agents occupying nodes of $\Delta_v$ have utility $0$. Observe that if there exist two agents of the same type or $t_1 > 0$, then some agent occupying a node of $\Delta_v$ would deviate to $v$ to increase her utility from $0$ to positive. Hence, all such agents must be of different type and $t_1 = 0$. But then, one of these agents has again incentive to deviate to $v$ to connect to the agent that occupies node $\beta$.

\item $v \in \Delta$. Suppose that the agent $i$ who occupies the parent of $v$ is of type $T_\ell$. For such an assignment to be an equilibrium, it must be the case that all other agents of type $T_\ell$ have utility $1$, as otherwise they would have incentive to deviate to $v$. Since $t_1 < 1$, this would mean that all other agents of type $T_\ell$ have only neighbors of type $T_\ell$; it is easy to see that this is impossible. 
\end{itemize}
This completes the proof.
\end{proof}

Since Theorem~\ref{thm:inexistence} implies that it is impossible to hope for general positive existence results, in the remainder of this section we focus on games with structured topologies and tolerance vectors. In particular, we consider the class of $\alpha$-binary $\lambda$-TS games with $\alpha \in \{1, \ldots, \lambda\}$ in which the tolerance vector $\tol{\lambda}$ is such that 
\begin{align*}
t_d =
\begin{cases}
1, & \text{if } d < \alpha \\
0, & \text{otherwise.}
\end{cases}
\end{align*}
Clearly, the class of $1$-binary $\lambda$-TS games coincides with that of $\lambda$-ZTS.

We next show that when the topology is a grid{\footnote{We focus on $4$-grids where internal nodes have $4$ neighbors.} or a tree, there exist values of $\alpha \in \{1, \ldots, \lambda\}$ for which $\alpha$-binary $\lambda$-TS games played on such a topology always admit at least one equilibrium. Our first result for grids is the following. 

\begin{theorem}\label{thm:grids-2types}
Every $2$-ZTS game $\calI = (N,G)$ in which $G$ is a grid admits at least one equilibrium.
\end{theorem}

\begin{proof}
Consider an arbitrary $2$-ZTS game $\calI = (N,G)$ in which the topology $G$ is an $m \times M$ grid ($m$ rows and $M$ columns), such that $m \leq M$. Also, let $x \geq 2$ be the number of agents per type, such that $n = 2x$. 
Clearly, since $n\geq 4$ it holds $M\geq 3$. To simplify our discussion, we will refer to the two types as red and blue.

We construct an equilibrium by assigning the agents to the nodes of the grid column-wise from top to bottom and left to right as follows: We first assign all the red agents. We then leave a number $e =  m\cdot M - n \geq 1$ of nodes empty, and finish by assigning all the blue agents. Let $\vv$ be the assignment computed.

Clearly, if $e \geq m$, $\vv$ is guaranteed to be an equilibrium since all agents have maximum utility $1$. Otherwise, when $e<m$ and since $M\geq 3$, it cannot be that there are $x \leq m$ agents per type. So, in the following we focus on the case where $e < m$ and $x> m$.

Let $v$ be an empty node, and consider any red agent $i$ that occupies node $v_i$ according to $\vv$. We will argue that $i$ has no incentive to jump to $v$; the argument for blue agents is symmetric. Observe that agent $i$ has utility at least $1/2$, since she is connected to at least one red agent and at most one blue agent by construction. We will focus on the case where agent $i$ has exactly one blue neighbor, as otherwise she clearly has no incentive to jump to $v$. By the construction of $\vv$ and the fact that $e <m$, we also have that the left neighbor of $v$ is occupied by a red agent, and the right neighbor of $v$ is occupied by a blue agent.
Moreover, the top neighbor of $v$, call it $\alpha$, can either be empty or occupied by a red agent, and  the bottom neighbor of $v$, call it $\beta$, can either be empty or occupied by a blue agent. We distinguish between the following cases:
\begin{itemize}
\item Case I: $\alpha$, $\beta$ are both empty, or $\alpha$ is empty and $\beta$ is occupied by a blue agent, or $\alpha$ is occupied by a red agent and $\beta$ is occupied by a blue agent. Agent $i$ has no incentive to deviate to $v$ since by doing so she would only be able to get utility at most $1/2$.

\item Case II: $\alpha$ is occupied by a red agent and $\beta$ is empty. If agent $i$ occupies $\alpha$ she has no incentive to deviate to $v$, as she would be able to get utility exactly $1/2$, while right now she has utility at least $1/2$. As a result, we have that agent $i$ is not adjacent to any empty node, and by deviating to $v$ agent $i$ would obtain utility $2/3$. The only case in which agent $i$ would have incentive to deviate to $v$ is if her utility is exactly $1/2$, which would mean that she has exactly one red neighbor and one blue neighbor. This is possible only if $v_i$ is the last node of the first column of the grid and $M=3$. However, this contradicts the fact that $m \leq M$, since  there is a column with two empty nodes ($v$ and $\beta$), one node ($\alpha$) occupied by a red agent, and one node (the right neighbor of $v_i$) occupied by a blue agent.
\end{itemize}
This completes the proof.
\end{proof}

The proof of Theorem~\ref{thm:grids-2types} is constructive and such that in the computed equilibrium no agent is isolated. Consequently, in combination with Theorem~\ref{thm:2-type-equivalence}, it further implies the following:
\begin{cor}
Every $2$-TS game $\calI = (N,G,\tol{\lambda})$ in which $G$ is a grid admits at least one equilibrium.
\end{cor}
Unfortunately, showing a result similar to  Theorem~\ref{thm:grids-2types} for every $\lambda \geq 3$ is a very challenging task. Instead, we  show the following result for $2$-binary games.

\begin{theorem} \label{thm:2-binary-grid}
Every $2$-binary $\lambda$-TS game $\calI = (N,G,\tol{\lambda})$ in which $G$ is a grid admits at least one equilibrium.
\end{theorem}

\begin{proof}
Consider a $2$-binary $\lambda$-TS game with $n$ agents played on an $m \times M$ grid ($m$ rows and $M$ columns) such that $m \leq M$.
Let $x=n/\lambda \geq 2$ be the number of agents per type and $e = mM - n$ be the number of empty nodes.

We compute an equilibrium assignment using Algorithm \ref{alg:binary-grid}, which in turn relies on the \textsc{Tile} procedure described in Algorithm \ref{alg:tile}. In particular, Algorithm \ref{alg:tile} takes as input the set of yet unassigned agents, an $r \times M$ subgrid and an integer $k\leq M$, and outputs an assignment of agents to the nodes in the $r \times M$ subgrid so that the leftmost $k$ nodes in the top row are empty, while all other nodes host an agent if the number of unassigned agents is large enough. Algorithm \ref{alg:tile} considers the yet unassigned agents in increasing type according to the ordering $\succ$ and assigns them in these rows one after the other, from row $1$ to row $r$ along the columns, from left to right, skipping the $k$ empty nodes.

\begin{algorithm}[!t]
\caption{Equilibrium construction for a $2$-binary $\lambda$-TS game on an $m\times M$ grid}\label{alg:binary-grid}
\DontPrintSemicolon
\tcc{$x$: number of agents per type}
\tcc{$e$: number of empty nodes}
\tcc{The algorithm terminates immediately when all agents have been assigned. It processes groups of rows by calling Algorithm \ref{alg:tile}.}
 \While {$x \leq m$ and $e\geq M$}{
 \textsc{Tile}($x,0$)\;
leave the next row empty\;
update $m := m - x - 1$, $e := e - M$
}
\If {$x >m$}{
\textsc{Tile}($m,0$) 
}
\Else(\tcc*[h]{In this case it holds that $e<M$ and $x\leq m$}){
Define non-negative integers $\alpha \in \mathbb{N}_{>0}$ and $\beta\leq x-1$ such that $m=\alpha x+\beta$

\For {$i=1, \dots, \alpha-1$}{
\textsc{Tile}($x,0$) 
}
\If {$\beta=0$}{
\textsc{Tile}($x,e$) 
}
\ElseIf {$\beta=1$}{
\textsc{Tile}($1,e$)\; 
\textsc{Tile}($x,0$) 
}
\Else{\textsc{Tile}($x,0$) \;
\textsc{Tile}($\beta,e$) 
}
}
\end{algorithm}

\LinesNumberedHidden
\begin{algorithm}[!t]
\caption{\textsc{Tile}($r,k$)}\label{alg:tile}
\DontPrintSemicolon
\tcc{$r$: number of rows, definining an $r \times M$ subgrid}
\tcc{$k$: number of nodes to be left empty}
\For{$i = 1$ to $k$}{
mark node $(1,i)$ as empty}
\For{$j = 1$ to $M$} { 
 \For{$i = 1$ to $r$}{
\If {node $(i,j)$ is unmarked}{place the next agent (if one exists) according to the ordering $\succ$ at node $(i,j)$}
}}
\end{algorithm}

Algorithm \ref{alg:binary-grid} terminates immediately (at any step) when all agents have been assigned.  We now argue that the resulting allocation $\vv$ computed by Algorithm \ref{alg:binary-grid} is an equilibrium.  Observe that all agents assigned at Lines 2 and 6 get utility $1$, as each agent of type $\ell$ has neighbors of types in $\{\ell-1, \ell, \ell+1\}$. Clearly, if the algorithm terminates at these lines (i.e., all agents are placed),
$\vv$ is an equilibrium.

Note that the algorithm cannot terminate at Lines 9--10 since $e<M$, while each agent $i$ of type $\ell$ placed during this step has utility at least $2/3$; indeed, $i$ has at least one neighbor of type $\ell$, at least one neighbor of a type in $\{\ell-1, \ell+1\}$, and at most one neighbor of type at distance at least $2$.

In Line 12, agents placed at the last $x-1$ rows have utility $1$, while an agent $i$ of type $\ell$ on the row with the empty nodes  has utility at least $1/2$ when $e=M-1$, since she has at least one neighbor of type $\ell$ and  at most one neighbor of type at distance at least $2$, and at least $2/3$ otherwise, since she has at least one neighbor of type $\ell$, at least one neighbor of a type in $\{\ell-1, \ell+1\}$, and at most one neighbor of type at distance at least $2$.

In Line 14, agents assigned have utility at least $2/3$ except (perhaps) the first or the last agent on the row that has utility at least $1/3$, since $x\geq 2$. All agents placed in Lines 15 and 17 have utility at least $2/3$ with a similar reasoning as above. In Line 18, an agent $i$ of type $\ell$  has utility at least $1/2$ when $e=M-1$, since she has at least one neighbor of type $\ell$ and at most one neighbor of type at distance at least $2$, and at least $2/3$ otherwise, since she has at least two neighbors of a type in  $\{\ell-1,\ell,\ell+1\}$
and at most one neighbor of type at distance at least $2$. 

Observe that, in any case, an agent gets utility at most $2/3$ by jumping to an empty node, since either the top or the bottom neighbor will have a large type distance and there is no left neighbor. As in almost all cases, agents in $\vv$ have utility at least $2/3$, it remains to argue about the nodes that have utility less than that. The agent in Line 12 with utility $1/2$ obtains utility at most $1/2$ by jumping, the agents in Line 14 with utility at least $1/3$ obtain utility at most $1/3$ by jumping and, finally, the agent in Line 18 with utility $1/2$ obtains utility at most $1/2$ by jumping. We conclude that $\vv$ is an equilibrium and the theorem follows.
\end{proof}


Note that Algorithm \ref{alg:binary-grid} may fail to return an equilibrium for lexicographically larger tolerance vectors. Indeed, consider a $4\times 4$ grid and $7$ types of two agents each. Algorithm \ref{alg:binary-grid} puts agents of types $1$ to $4$ in each of the first two rows, skips $2$ nodes, puts agents of types $6$ and $7$ in the third row, and places agents of types $5$, $5$, $6$ and $7$ in the last row; see also Figure \ref{fig:tile-example}. Under tolerance
vector $\tol{7} = \{1, 1, 0, 0, 0, 0, 0\}$ the assignment is an equilibrium (by Theorem \ref{thm:2-binary-grid}), while under tolerance vector $ \mathbf{t'_{7}}= \{1, 1, t_2>\frac{1}{2}, 0, 0, 0\}$, the agent of type $4$ in the second row has utility $2/3$, but can obtain utility $\frac{1+2t_2}{3} > 2/3$ by jumping to the rightmost empty node.

\begin{figure}[ht]
\centering
\includegraphics[scale=0.55]{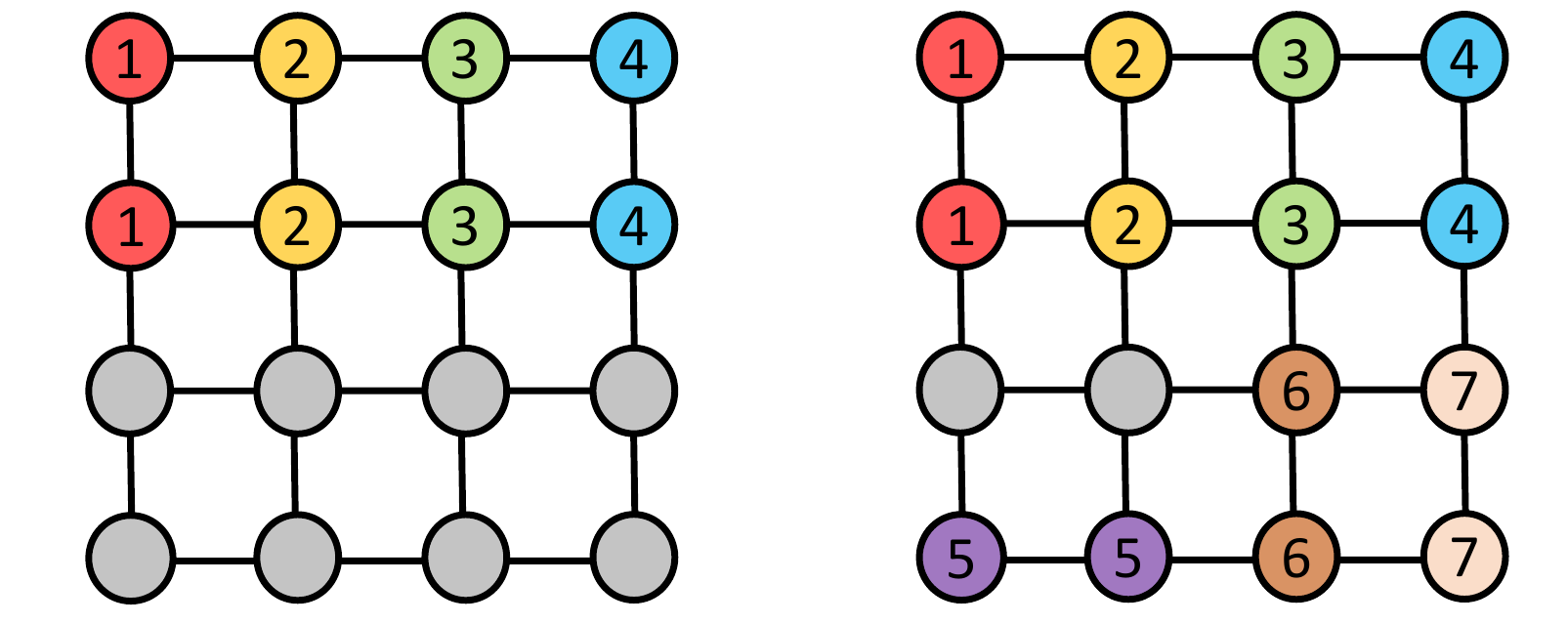}
\caption{An example of how Algorithm \ref{alg:binary-grid} operates. On the left, we see the grid after the for loop in Lines 9--10, while on the right we see the final outcome after Line 12 is executed. Agents of the same color and number are of the same type.}
\label{fig:tile-example}
\end{figure}

So, a different proof is needed for computing equilibria in $\alpha$-binary games with $\alpha \geq 3$. While we have not been able to show this result for every $\alpha$, we do show it for $\alpha\geq \sqrt{\lambda}$. In particular, the equilibrium constructed in the proof of the next theorem guarantees a utility of $1$ to all agents, and thus it is also an equilibrium for games with lexicographically larger tolerance vectors, not necessarily binary ones.

\begin{theorem}\label{thm:grids-3types-rootk}
For $\lambda \geq 3$, every $\sqrt{\lambda}$-binary $\lambda$-TS game $\calI = (N,G,\tol{\lambda})$ in which $G$ is a grid admits at least one equilibrium.
\end{theorem}

\begin{proof}
Consider a $\sqrt{\lambda}$-binary $\lambda$-TS game $\calI=(N,G,\tol{\lambda})$ in which the topology $G$ is an $m \times M$ grid ($m$ rows and $M$ columns) such that $m \leq M$. Observe that if an agent $i$ is assigned to a node $v_i$ such that {\em all} her neighbors are of types in distance strictly less than $\sqrt{\lambda}$ according to the ordering $\succ$, then agent $i$ has maximum utility $1$, and no incentive to deviate from $v_i$. This is the property we will exploit to construct an equilibrium assignment.
We distinguish between the following two cases:
\begin{itemize}
\item $m \leq \sqrt{n}$. We construct an equilibrium assignment $\vv$ by considering the agents in increasing type according to the ordering $\succ$, and assign them one after the other along the columns, from top to bottom and from left to right. Since there are at most $\sqrt{n}$ rows and $\sqrt{n} = \frac{n}{\lambda} \cdot \frac{\lambda}{\sqrt{n}}$, each column of the grid can fit at most all the agents of $\frac{\lambda}{\sqrt{n}}$ different types. This means that, by the construction of $\vv$, every agent $i$ of type $\ell$ has neighbors of types $T_k$ such that $|\ell - k| \leq \frac{\lambda}{\sqrt{n}} < \sqrt{\lambda}$, where the last inequality follows by the fact that $n \geq 2\lambda$. Hence, agent $i$ has utility $1$ and no incentive to deviate.

\item $m > \sqrt{n}$. Consider the sub-grid containing only the nodes in the first $\sqrt{n}$ rows. Since $M \geq m > \sqrt{n}$, all $n$ agents can be assigned to these nodes, and we thus can repeat the process of the previous case by limiting our attention to the sub-grid. This again guarantees that every agent has utility $1$ and the resulting assignment is an equilibrium. 
\end{itemize} 
This completes the proof.
\end{proof}

Next we turn our attention to games in which the topology is a tree. We show the following result for $\alpha$-binary games when $\lambda \geq 3$.

\begin{theorem}\label{thm:trees-gen}
Every $2$-binary $3$-TS game $\calI=(N,G,\tol{3})$ and every $\alpha$-binary $\lambda$-TS game $\calI=(N,G,\tol{\lambda})$ where $\alpha\geq \left\lfloor{\frac{\lambda}{2}}\right\rfloor$ for $\lambda \geq 4$, in which $G$ is a tree, admit at least one equilibrium. 
\end{theorem}

\begin{proof}
To construct an equilibrium, we exploit the following known property of trees: Every tree with $x\geq 3$ nodes contains a {\em centroid} node, whose removal splits the tree into at least two subtrees with at most $x/2$ nodes each. We root the tree from such a centroid node, and leave the root empty. 
This leads to a partition of the topology in $k \geq 2$ subtrees, which we order in non-increasing size and denote by $tree_1, \ldots, tree_k$.

To assign the agents we use Algorithm~\ref{alg:tree}, which in turn uses the {\textsc{Bottom-Up}} allocation procedure (described in Algorithm~\ref{alg:treesk}). The procedure {\textsc{Bottom-Up}}($tree, \mathcal{T}_1,\mathcal{T}_{2},\ldots,\mathcal{T}_{s}$) assigns the unassigned agents of types $\mathcal{T}_1,\mathcal{T}_{2}, \ldots,\mathcal{T}_{s}$ to the nodes of the subtree $tree$ from bottom to top (higher to lower depth), so that all the agents of $\mathcal{T}_1$ are covered by either agents of the same type or agents of type $\mathcal{T}_2$, and the assignment for the remaining agents is connected.

\begin{algorithm}[!t]
\caption{Equilibrium construction for a $\left\lfloor{\frac{\lambda}{2}}\right\rfloor$-binary $\lambda$-TS game on a tree (or games with lexicographically larger tolerance vectors)}\label{alg:tree}
\DontPrintSemicolon

\tcc{$tree_1, \ldots, tree_k$ denote the subtrees of the tree topology in non-increasing order by size, when the topology is rooted at a centroid node.}

Run {\textsc{Bottom-Up}}($tree_1, T_1,T_{2},\ldots,T_{\left\lceil\frac{\lambda}{2}\right\rceil}$).
If at least one agent of type $T_1$ remains unassigned, repeat with the next subtree. 
Let $a \leq \left\lceil\frac{\lambda}{2}\right\rceil$ be the smallest type index among unassigned agents, and let $tree_{k_1}$ be the last subtree considered in this step.

Run {\textsc{Bottom-Up}}($tree_{k_1+1}, \mathcal{T}_\lambda,\mathcal{T}_{\lambda-1},\ldots,\mathcal{T}_{\left\lceil\frac{\lambda+1}{2}\right\rceil}$), where $\mathcal{T}_i$ are the unassigned agents of types $T_i$, $i=\lambda,\ldots, \left\lceil\frac{\lambda+1}{2}\right\rceil$.
If at least one agent of type $T_\lambda$ remains unassigned, repeat with the next subtree. 
Let $b \geq \left\lceil\frac{\lambda+1}{2}\right\rceil$ be the largest type index among unassigned agents, and let $tree_{k_2}$ be the last subtree considered in this step.

Run {\textsc{Bottom-Up}}($tree_{k_2+1}, \mathcal{T}_a,\mathcal{T}_{a+1},\ldots,\mathcal{T}_{b}$), where $\mathcal{T}_i$ are the unassigned agents of types $T_i$, $i=a,\ldots, b$.
Repeat with the next subtree and the unassigned agents of these types, until all agents have been assigned.

If the last subtree among the ones considered in the previous steps contains at least two isolated agents, then rearrange them within this subtree so that each of them has at least one neighbor. If the last subtree contains a single isolated agent, then move this agent to the root of the tree.
\end{algorithm}

\begin{algorithm}[!t]
\caption{{\textsc{Bottom-Up}}($tree, \mathcal{T}_1,\mathcal{T}_{2},\ldots,\mathcal{T}_{s}$)}\label{alg:treesk}
\DontPrintSemicolon

\tcc{For $i=1,\ldots,s$, $\mathcal{T}_i$ is the set of unassigned agents of a given type}
\tcc{The algorithm terminates immediately when all agents have been assigned or all nodes of $tree$ have been occupied.}

Start at the lowest level of $ tree$ and place agents of type $\mathcal{T}_1$ so that an agent of type $\mathcal{T}_1$ is placed at level $h$ only if all nodes at levels at least $h+1$ have been filled. Furthermore, and assuming the previous condition holds, after filling a node at level $h$ we give priority to its sibling nodes. Continue until all agents of type $\mathcal{T}_1$ have been assigned.

Consider the agents of type $\mathcal{T}_2$. Begin by placing an agent of type $\mathcal{T}_2$ to any empty node having a child occupied by an agent of type $\mathcal{T}_1$ and repeat until the parent nodes of all agents of type $\mathcal{T}_1$ are occupied. This is feasible as long as there are at least as many agents of type $\mathcal{T}_2$ as there are agents of type $\mathcal{T}_1$. Continue by placing agents of type $\mathcal{T}_2$ arbitrarily in $tree$ by maintaining a connected assignment. 

Arbitrarily assign the remaining agents in order of input so that the assignment remains connected after assigning each agent.
\end{algorithm}


We first claim that at the end of Step 3 of Algorithm \ref{alg:tree}, every agent either gets utility $1$ or gets utility $0$ if she is isolated. Indeed, it holds that agents of type $T_1$ can only be adjacent to agents of type $T_1$ and $T_2$. Similarly, the agents of type $T_\lambda$ can only be adjacent to agents of type $T_\lambda$ and $T_{\lambda-1}$. In addition, by design, the maximum type distance among all the other agents assigned in Steps 1 and 2 is $\left\lceil\frac{\lambda}{2}\right\rceil-2$. By this discussion, all agents have utility $1$ when $\lambda = 3$ and the game is $2$-binary. Below, we assume that $\lambda \geq 4$. 

To see the claim is true for agents assigned in Step 3, observe that if Step 2 is applied on a subtree of at least $n/3$ nodes, then since we visit subtrees in non-increasing order of their size, Step 1 is also applied on a subtree of at least $n/3$ nodes. Hence, at most $n/3$ agents remain to be allocated. 
Otherwise, if no subtree on which Step 2 is applied has at least $n/3$ nodes, then, again due to the order we visit subtrees, any subtree to which we perform Step 3 has less than $n/3$ nodes. In any case, at most $n/3$ agents will be allocated at Step 3 at any given subtree. These agents belong to at most
$\lceil \lambda/3\rceil+1$ different types and, due to Steps 2 and 3 in Algorithm \ref{alg:treesk}, we are guaranteed that no agent allocated in Step 3 will have a neighbor of type-distance $\lceil \lambda/3 \rceil$.
Since $\lceil \lambda/3\rceil -1 \leq \lfloor \lambda/2\rfloor -1$, such agents either get utility $1$, or $0$ if they are isolated, as required.

It remains to argue that after a possible execution of Step 4, no agent has a profitable deviation.
We distinguish between the following two cases when Step 4 is performed:
\begin{itemize}
\item Case I: There are at least two isolated agents in the last subtree among those considered in the first three steps. First observe that, since the subtrees are considered in non-increasing order by size and the last subtree contains at least two agents, there is no subtree with a single isolated agent. Now, by the definition of the bottom-up-like allocation algorithm, all these agents must be of the last type $T_b$, since if agents of two or more types are assigned in the same subtree, the resulting assignment therein is by construction connected. Therefore, by rearranging the agents of type $T_b$ in the last subtree so that all of them have at least one neighbor, each of them gets utility $1$ and the assignment is an equilibrium.

\item Case II: There is a single isolated agent $i$ in the last subtree of the last type $T_b$ considered, who is moved to the root of the tree. Since Step 4 is performed, all the subtrees that have been considered in the first three steps are full, with the exception of the last subtree which has been left empty after moving agent $i$. Thus, the empty nodes of the topology are only adjacent to other empty nodes or the root.  As a result, an agent of some type $\ell \in [\lambda]$ would be able to get utility $t_{|\ell-b|}$ by jumping to an empty node that is adjacent to the root, and utility $0$ by jumping to any other empty node. However, every agent $j \neq i$ already has  utility at least $t_{|\ell-b|}$. In particular, agent $j$ has utility $1$ if she is not adjacent to the root, utility at least $\frac{1+t_{|\ell-b|}}{2} \geq t_{|\ell-b|}$ if she is adjacent to the root but not isolated before moving $i$ to the root, and utility exactly $t_{|\ell-b|}$ if she is adjacent to the root and was isolated before moving $i$ to the root.
\end{itemize}
This completes the proof.
\end{proof}

For $\lambda = 3$, Theorem~\ref{thm:trees-gen} is tight in the sense that equilibria are not guaranteed to exist when $t_1 < 1$ (Theorem~\ref{thm:inexistence}). For $\lambda \geq 4$, it is not hard to observe that the assignment computed is also an equilibrium in games with lexicographically larger vectors (not necessarily binary ones) than the one stated. 


\section{Quality of equilibria}\label{sec:poa}

In this section, we consider the quality of equilibria measured in terms of social welfare, and bound the price of anarchy and price of stability. Recall that these notions compare the social welfare achieved in the {\em worst} and {\em best} equilibrium to the maximum possible social welfare achieved in any assignment. We start with a general upper bound on the price of anarchy, whose proof follows by bounding the social welfare at equilibrium by the total utility the agents would be able to obtain by jumping to an arbitrary empty node.
Recall that $\tau = \sum_{d=0}^{\lambda-1} t_d$.

\begin{theorem}\label{thm:poa-upper-general}
The price of anarchy of $\lambda$-TS games with tolerance vector $\tol{\lambda}$ is at most
$\frac{\lambda n}{\tau n - \lambda }.$
\end{theorem}

\begin{proof}
Consider a $\lambda$-TS game $\calI = (N,G,\tol{\lambda})$ with $\EQ(\calI) \neq \varnothing$. Let $\vv$ be an equilibrium, and denote by $v$ an empty node. The utility that an agent of type $T_\ell$, $\ell \in [\lambda]$ would  obtain by unilaterally jumping to $v$ is
\begin{itemize}
\item $\frac{1}{n(v)} \sum_{k \in [\lambda]} t_{|\ell - k|} \cdot n_k(v)$ if she is not adjacent to $v$;
\item $\frac{1}{n(v)-1} \left( \sum_{k \in [\lambda]} t_{|\ell-k|} \cdot n_k(v) - 1 \right)$ otherwise.
\end{itemize}
Also observe that for every type $T_\ell$, $\ell \in [\lambda]$ there are exactly $\frac{n}{\lambda}- n_\ell(v)$ agents that are not adjacent to $v$, and $n_\ell(v)$ agents that are adjacent to $v$. Since $\vv$ is an equilibrium, every agent of type $T_\ell$ is guaranteed  to have at least as much utility as if she were to deviate to $v$, and therefore the social welfare is
\begin{align*}
\SW(\vv)
&\geq
\frac{1}{n(v)} \sum_{\ell \in [\lambda]} \left( \frac{n}{\lambda}-n_\ell(v) \right) \sum_{k \in [\lambda]} t_{|\ell-k|} \cdot n_k(v) 
 + \frac{1}{n(v)-1} \sum_{\ell \in [\lambda]} n_\ell(v) \cdot \left( \sum_{k \in [\lambda]}  t_{|\ell-k|} \cdot n_k(v) - 1 \right) \\
&\geq \frac{1}{n(v)} \sum_{\ell \in [\lambda]} \left( \frac{n}{\lambda}  \sum_{k \in [\lambda]} t_{|\ell-k|} \cdot n_k(v) - n_\ell(v)  \right) \\
&= \frac{1}{n(v)} \sum_{\ell \in [\lambda]} n_\ell(v) \cdot \left( \frac{n}{\lambda} \sum_{k \in [\lambda]} t_{|\ell-k|} - 1 \right) \\
&= \frac{1}{\lambda \cdot n(v)}  \sum_{\ell \in [\lambda]} n_\ell(v) \left( n \sum_{k \in [\lambda]} t_{|k-\ell|} - \lambda \right).
\end{align*}
The second inequality is due to increasing the denominator of the second fraction. The first equality follows by aggregating the factors of $n_\ell(v)$ for every $\ell \in [\lambda]$. Finally, the second equality follows by factorizing $\lambda$. Now observe that because the tolerance vector $\tol{\lambda}$ is non-increasing, we have that $\sum_{k \in [\lambda]} t_{|\ell-k|} \geq \sum_{d=0}^{\lambda-1} t_d = \tau$. Combining this together with the fact that $n(v) = \sum_{\ell \in [\lambda]}n_\ell(v)$, we obtain
\begin{align*}
\SW(\vv) &\geq \frac{ \tau n  - \lambda }{\lambda}.
\end{align*}
The bound on the price of anarchy follows by the fact that the optimal welfare is at most $n$ (the maximum utility of any agent is $1$).
\end{proof}

For each $\ell \in \{1, \ldots, \lambda\}$, let $\tau_\ell = \sum_{k \in [\lambda]} t_{|\ell-k|}$ be the total tolerance of agents of type $\ell$ towards any subset containing one agent of every type. We can show the following general lower bound on the price of anarchy, as a function of these parameters.

\begin{theorem}\label{thm:poa-lower-bound}
The price of anarchy of $\lambda$-TS games with tolerance vector $\tol{\lambda}$ is at least $$\frac{\lambda n}{\frac{\sum_{\ell \in [\lambda]}{\tau_\ell}}{\lambda}n-\frac{\lambda^2-\sum_{\ell \in [\lambda]}{\tau_\ell}}{\lambda-1}}\geq 
\frac{\lambda n}{\frac{2(\lambda-1)\tau}{\lambda}n-\frac{\lambda^2}{\lambda-1}+2\tau}.$$
\end{theorem}

\begin{proof}
Consider a $\lambda$-TS game $\calI = (N,G,\tol{\lambda})$, where $n = |N| = \lambda(2\lambda\mu+1)$ for some integer $\mu\geq 1$.
The topology $G$ contains $2n+1$ nodes as follows. There exists a central node $c$ and $\lambda$ large cliques $K_i$, with $i \in [\lambda]$, where each $K_i$ contains $2\lambda\mu+1$ nodes. There also exist $2\lambda\mu$ small cliques $K_{i,j}$, with $i\in [\lambda]$ and $j\in\{1,\dots,2\mu\}$, each of size $\lambda$, as well as another small clique $K$ of size $\lambda$. Apart from the edges within the cliques, we add the following edges so that any node of a small clique is connected to a single node outside its clique. First, we connect the central node $c$ to a single node of cliques $K_i$ and $K_{i,1}$, for $i \in [\lambda]$. Then, for each $K_{i,1}$, with $i\in[\lambda]$, we connect each node that is not a neighbor of $c$ to a single distinct node in $K_{i,2}$. Similarly, for each $K_{i,j}$, with $i\in\{2, \dots, \lambda\}$ and $j\in\{1, \dots,2\mu-1\}$, we connect any node without a neighbor in $K_{i,j-1}$ to a single distinct node in $K_{i,j+1}$. Finally, for each $K_{i,2\mu}$, with $i\in\{1,\dots,2\mu\}$, we connect the only node with no neighbors in $K_{i,2\mu-1}$ to a single distinct node in $K$; this completes the description of $G$.

In the optimal assignment, for $\ell \in [\lambda]$, all agents of type $\ell$ are at the large clique $K_\ell$ and $\OPT(\calI) = n$. There exists an equilibrium assignment $\vv$ where all agents are placed in the small cliques as follows: each small clique contains a single agent of each type,  the agents neighboring the (empty) central node are all of distinct types, and, finally, for each pair of neighboring agents across different small cliques, both agents are of the same type. It is always possible to obtain such an assignment; see Figure \ref{fig:poa-lower-balanced} for an example with $\lambda = 3$ and $\mu=1$.

\begin{figure}[t]
\centering
\includegraphics[scale=0.55]{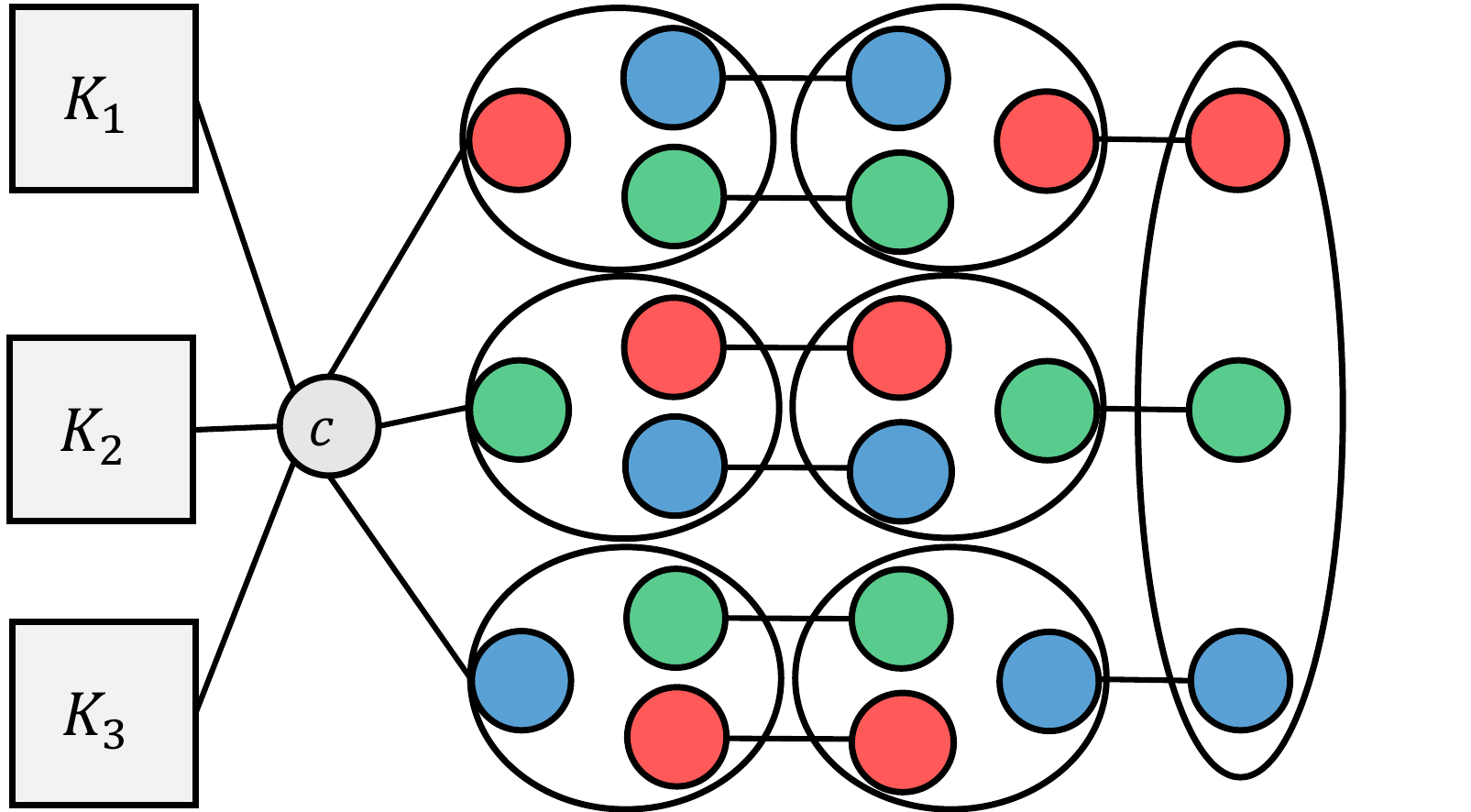}
\caption{An instance used for the proof of Theorem~\ref{thm:poa-lower-bound} for the case of $3$ types. There are $21$ agents, so that each type has $7$ agents. The big squares $K_1$, $K_2$, $K_3$ correspond to cliques of size $7$ (the number of agents per type), while the ovals represent cliques of size $3$ (the number of types); intra-clique edges are omitted.  In an optimal assignment, each large clique contains agents of the same type and each agent gets utility $1$. In a bad equilibrium, each small clique contains a single agent of each type and all gray nodes are left empty. For each type $\ell\in[3]$, all but one agents of type $\ell$ get utility $\tau_\ell/3$, while the last agent gets utility $({\tau_\ell-1})/2$.}
\label{fig:poa-lower-balanced}
\end{figure}

Observe that, for $\ell \in [\lambda]$,  each agent $i$ of type $\ell$ that is not a neighbor of the central node has utility $u_i(\vv) = \frac{\tau_\ell}{\lambda}$ and will obtain the same utility if she jumps to the central node. Similarly, each agent $i$ of type $\ell$ who is a neighbor of $c$ has utility $\frac{\tau_\ell-1}{\lambda-1}$ and will get the same utility by jumping to the central node. Clearly, any agent will get utility $0$ by jumping to another empty node. Hence, $\vv$ is an equilibrium with
\begin{align*}
\SW(\vv) &= \sum_{\ell\in[\lambda]}{\left(2\lambda\mu \frac{\tau_\ell}{\lambda}+\frac{\tau_\ell-1}{\lambda-1}\right)}\\
&= \sum_{\ell\in[\lambda]}{\left(\left(\frac{n}{\lambda}-1\right)\frac{\tau_\ell}{\lambda}+\frac{\tau_\ell-1}{\lambda-1}\right)}\\
&= \frac{\sum_{\ell \in [\lambda]}{\tau_\ell}}{\lambda^2}n+\frac{\sum_{\ell \in [\lambda]}{\tau_\ell-\lambda^2}}{\lambda(\lambda-1)}.
\end{align*}
Since $\OPT(\calI) = n$, we obtain the following lower bound on the price of anarchy:
$$\frac{\lambda n}{\frac{\sum_{\ell \in [\lambda]}{\tau_\ell}}{\lambda}n-\frac{\lambda^2-\sum_{\ell \in [\lambda]}{\tau_\ell}}{\lambda-1}}.$$
Furthermore, since $\tau_\ell < 2\tau$ for every $\ell \in \{2, \dots, \lambda-1\}$ and $\tau = \tau_1 = \tau_\lambda$, it holds that $\sum_{\ell \in [\lambda]}{\tau_\ell} \leq 2(\lambda-1)\tau$. Hence, we have a lower bound of
$$\frac{\lambda n}{\frac{\sum_{\ell \in [\lambda]}{\tau_\ell}}{\lambda}n-\frac{\lambda^2-\sum_{\ell \in [\lambda]}{\tau_\ell}}{\lambda-1}} \geq \frac{\lambda n}{\frac{2(\lambda-1)\tau}{\lambda}n-\frac{\lambda^2}{\lambda-1}+2\tau},$$
as desired.
\end{proof}

From Theorems~\ref{thm:poa-upper-general} and ~\ref{thm:poa-lower-bound} we obtain an asymptotically tight bound for general $\lambda$-TS games.

\begin{cor} 
The price of anarchy in $\lambda$-TS games is $\Theta(\lambda/\tau)$.
\end{cor}

Theorem \ref{thm:poa-lower-bound} allows us to provide concrete bounds for subclasses of $\lambda$-TS games. In particular, for $\lambda$-ZTS games, since $\tau_\ell =1$ for every $\ell \in [\lambda]$, we have $\sum_{\ell \in [\lambda]}{\tau_\ell} = \lambda$, and thus the left-hand-side of the inequality in Theorem \ref{thm:poa-lower-bound} yields the following tight bound.

\begin{cor}
The price of anarchy of $\lambda$-ZTS games is $\frac{\lambda n}{n-\lambda}$.
\end{cor}

We now define the following two natural classes of $\lambda$-TS games in which the tolerance parameters are specific functions of the distance between the types.
\begin{itemize}
\item {\em Proportional} $\lambda$-TS games: $t_d = 1 - \frac{d}{\lambda-1}$  for every  $d \in \{0, \ldots, \lambda-1\}$. We have
$\tau = \sum_{\ell \in [\lambda]} \frac{\ell-1}{\lambda-1} = \frac{\lambda}{2}.$
\item {\em Inversely proportional} $\lambda$-TS games: $t_d = \frac{1}{d+1}$ for every $d \in \{0, \ldots, \lambda-1\}$. We have
$\tau = \sum_{\ell\in [\lambda]}{\frac{1}{\ell}}=H_\lambda,$
where $H_\lambda$ is the $\lambda$-th harmonic number.
\end{itemize}
By Theorems~\ref{thm:poa-upper-general}, ~\ref{thm:poa-lower-bound} and the above definitions, we obtain the following corollaries. 

\begin{cor}
For every $\lambda \geq 2$, 
the price of anarchy of proportional $\lambda$-TS games is 
at most $\frac{2n}{n-2}$ and 
at least $\frac{\lambda n}{(\lambda-1)n-\frac{\lambda}{\lambda-1}}$.
\end{cor}

\begin{cor}
For every $\lambda \geq 2$, 
the price of anarchy of inversely proportional $\lambda$-TS games is 
at most $\frac{\lambda n}{H_\lambda n-\lambda}$ and 
at least $\frac{\lambda n}{\frac{2(\lambda-1)}{\lambda}H_\lambda n-\frac{\lambda^2}{\lambda-1}+2 H_\lambda}$.
\end{cor}

We conclude our technical contribution with a lower bound on the price of stability for the case of two types of agents. For $2$-ZTS games, the following lower bound improves upon the bound of $34/33$ of \citet{elkind2019jump}, and is also tight when the number of agents tends to infinity because of the upper bound implied by Theorem~\ref{thm:poa-upper-general}; recall that $\tau=1$ for $\lambda$-ZTS games.

\begin{theorem}\label{thm:KK-PoS-jump}
The price of stability of $2$-TS games is at least $2/\tau-\epsilon$, for any $\epsilon>0$.
\end{theorem}

\begin{proof}
We consider a particular $2$-TS game $(N,G, \tol{2} = (1,0))$ and we will show that it admits a unique (up to symmetry) equilibrium $\vv$. Furthermore, we will show that $\vv$ has no isolated agents. By Theorem \ref{thm:2-type-equivalence}, the $2$-TS game $(N,G,\tol{2}' = (1,t_1<1))$ also admits $\vv$ as its unique equilibrium. It will then suffice to argue about the social welfare in $\vv$. Clearly, if $t_1=1$ then $\tau=2$ and the theorem holds trivially.

Let $b$ be an arbitrarily large even positive integer that is not a multiple of $3$ and set $z = 2b+1$ and $c = bz = 2b^2+b$. In our proof we exploit that $\gcd(cz/2+c, b+1)= \gcd(2b^3+4b^2+3b/2, b+1)=1$\footnote{Recall that $b$ is even and note that $2(2b^3+4b^2+3b/2)-(4b^2+4b-1)(b+1)=1$. The claim follows since the greatest common divisor of two integers equals their
smallest positive linear combination.}.

Let $n = |N| = 2c(z+1)$ and consider the following topology $G$ with $n+1$ nodes; see also Figure \ref{fig:pos-lower}.
The set of nodes comprises six sets $I$, $J$, $K$, $X$, $Y$, and $S$. Set $I$ contains $cz/2$ nodes, set $J$ contains $c$ nodes, while $K$ is a clique of $cz/2$ nodes. Set $X$ is a collection of $z$ subsets $X_1$, $X_2 \dots$, $X_z$, where each subset has $b$ nodes, while set $Y$ is a collection of $z$ subsets $Y_1$, $Y_2, \dots$, $Y_z$ where each subset contains $c$ nodes. Finally, set $S$ contains a single node $s$.

\begin{figure}[ht]
\centering
\includegraphics[scale=0.45]{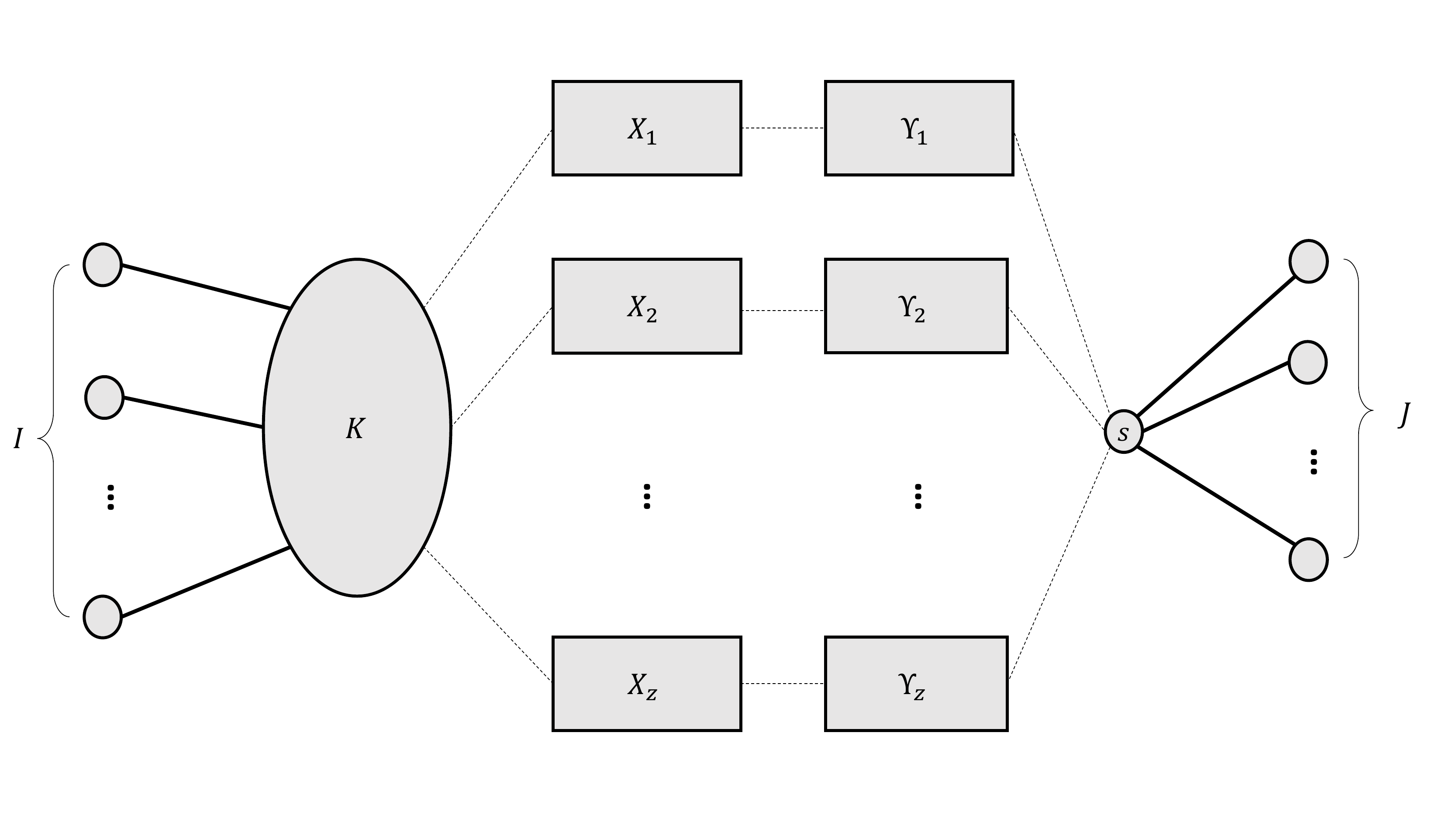}
\caption{The topology $G$ used in the proof. Small circles correspond to nodes, the oval is the clique $K$ while rectangles correspond to sets of nodes. Dashed edges connecting two shapes imply that each node of one shape is connected to all nodes of the other shape.}
\label{fig:pos-lower}
\end{figure}

Apart from the edges in the clique $K$, the following edges also exist. To begin with, each node in $I$ is connected to a unique distinct node in $K$. Each node in $K$ is also connected to all nodes in $X$, while each node in $X_i$ is also connected to all nodes in $Y_i$, for $i \in \{1, \dots, z\}$. Finally, node $s$ is connected to all nodes in $Y \cup J$.

Consider the following assignment $\vv$. All nodes in $I \cup K \cup X$ host red agents, all nodes in $S\cup J$ host blue agents, while all remaining $cz-1$ blue agents are in $Y$. We claim that $\vv$ is an equilibrium. To see why, observe that all agents in $I\cup K \cup S \cup J$ have utility $1$ and, hence, no incentive to jump to the empty node. Each (blue) agent in $Y$ has utility $\frac{1}{b+1}$ and would obtain the same utility by jumping. Finally, each (red) agent in $X$ has utility at least $\frac{z}{z+2}$ and would obtain utility at most $\frac{b}{b+1}$ by jumping; since $z>2b$, $\vv$ is indeed an equilibrium.

Next, we show that $\vv$ is the unique equilibrium assignment. This will be accomplished through a series of claims regarding the structure of any equilibrium assignment and the location of the unique empty node. In the following, let $r_K$ and $b_K$ denote the number of red and blue agents in $K$, $r_{x_j}$ and $b_{x_j}$ the number of red and blue agents in $X_j$ for $j\in\{1,\dots, z\}$, and, similarly, $r_{y_j}$ and $b_{y_j}$ the number of red and blue agents in $Y_j$ where $j \in \{1,\dots,z\}$. Finally, let $r_X$ and $b_X$ be the total number of red and blue agents in $X$.

\begin{claim} \label{cl:no-perfect}
In any equilibrium, there are at least two red and at least two blue agents with utility less than~$1$.
\end{claim}
\begin{proof}Observe that if $K\cup X$ contains at least two red and two blue agents, then the claim trivially holds. So, without loss of generality, let $K\cup X$ contain at most one red agent. Then, $Y$ must contain at least two red agents as the $cz/2+c+1$ nodes in $I\cup S\cup J$ cannot host all $cz+c-1$ remaining red agents. The claim follows by considering these red agents in $Y$ and their blue neighbors in $X$.
\end{proof}

\begin{claim}\label{cl:plokami}
In any equilibrium, the empty node cannot be in $I\cup J$.
\end{claim}
\begin{proof}Assume otherwise. Then, any agent having the same type as the unique neighbor of the empty node would obtain utility $1$ by jumping. Claim \ref{cl:no-perfect} guarantees that at least one such agent, other than the unique neighbor, exists.
\end{proof}

\begin{claim} \label{cl:perfect-plokami}
In any equilibrium, each agent in $J$ has utility~$1$.
\end{claim}
\begin{proof}Assume otherwise and let $i \in J$ be an agent with utility $0$; without loss of generality let us assume that $i$ is red. Since $i$ has no incentive to jump, it must be that all agents neighboring the empty node, except possibly for $i$, are blue. 
If $i$ is not a neighbor of the empty node, by Claim \ref{cl:no-perfect}, there exists at least one blue agent that has incentive to jump to the empty node and obtain utility $1$. So, we assume that $s$ is the empty node. In that case, all agents in $J$ have utility $0$ and it is not hard to see that such an assignment cannot be an equilibrium, since $J$ contains either at least two red agents or at least two blue agents. 
\end{proof}

This implies the following.
\begin{claim}\label{cl:s-cup-j} In any equilibrium, $s$ cannot be empty and all agents in $S\cup J$ are of the same type.
\end{claim}

\begin{claim} \label{cl:perfect-plokami-2}
In any equilibrium, each agent in $I$ has utility~$1$.
\end{claim}
\begin{proof} Assume otherwise and let $i \in I$ be an agent with utility $0$; without loss of generality let us assume that $i$ is red. As in the proof of Claim \ref{cl:perfect-plokami}, since $i$ has no incentive to jump, it must be that all agents neighboring the empty node, except possibly for $i$, are blue. 
If $i$ is not a neighbor of the empty node, by Claim \ref{cl:no-perfect}, there exists at least one blue agent that has incentive to jump to the empty node and obtain utility $1$. 
Otherwise, if the empty node is in $K$, then all agents in $K\cup X$ are blue. Furthermore, any other agent in $I$ except for $i$ is blue, as otherwise the agent would jump to the empty node and improve her utility from $0$ to strictly positive. We conclude that $I\cup K\cup X$ contains $cz-2$ blue agents, while the last two blue agents must be in $Y$, by Claim \ref{cl:s-cup-j}. The claim follows since each of these two blue agents has utility $\frac{b}{b+1}$ and would obtain utility $\frac{cz+c-2}{cz+c-1}$ by jumping; note that $cz+c-2>b$.
\end{proof}

By Claim \ref{cl:s-cup-j}, we assume, without loss of generality,  that all agents in $S\cup J$ are blue. This, together with Claim \ref{cl:perfect-plokami-2}, leads to:
\begin{align}\label{eq:red}
2r_K+r_X+r_Y = c(z+1),
\end{align}
and
\begin{align}\label{eq:blue}
2b_K+b_X+b_Y = cz-1.
\end{align}

The following claim relies on more complex arguments about almost the entire topology.

\begin{claim} In any equilibrium, the empty node cannot be in~$X$.
\end{claim}
\begin{proof}
Assume otherwise that the empty node is in $X_1$, without loss of generality. We begin by considering why agents in $Y\setminus Y_1$ have no incentive to jump. Let $i$ be an agent in $Y_j$, where $j \in \{2, \dots, z\}$ and observe that $r_K+b_K = cz/2$, $r_{x_j}+b_{x_j} = b$, and $r_{y_1}+b_{y_1} = c$.

If agent $i$ is red, her utility is $\frac{r_{x_j}}{b+1}$, as $s$ hosts a blue agent, while by jumping agent $i$ would get utility $\frac{r_K+r_{y_1}}{cz/2+c}$. Hence, since $i$ has no incentive to jump, we obtain $\frac{r_{x_j}}{b+1} \geq \frac{r_K+r_{y_1}}{cz/2+c}$ which gives
$\frac{b_{x_j}+1}{b+1} \leq \frac{b_K+b_{y_1}}{cz/2+c}$. Similarly, if agent $i$ is blue, we obtain $\frac{b_{x_j}+1}{b+1} \geq \frac{b_K+b_{y_1}}{cz/2+c}$ and, therefore, $\frac{r_{x_j}}{b+1} \leq \frac{r_K+r_{y_1}}{cz/2+c}$.

We conclude that if $Y_j$ contains both red and blue agents, it must hold that $\frac{r_{x_j}}{b+1} = \frac{r_K+r_{y_1}}{cz/2+c},$
and $\frac{b_{x_j}+1}{b+1} = \frac{b_K+b_{y_1}}{cz/2+c}$, which, as $\gcd(cz/2+c, b+1)=1$, is only possible when $b_{x_j}=b$, $b_K = cz/2$ and $b_{y_1}= c$. This, however, cannot hold as we already have that all $c+1$ agents in $S\cup J$ are blue and the remaining $cz-1$ blue agents cannot fill $x_j$, $y_1$, $K$ and, due to Claim \ref{cl:perfect-plokami}, $I$.

So, any $Y_j$ contains either only red or only blue agents. In the following, let us assume that $Y\setminus Y_1$ contains $\psi$ sets consisting only of blue agents and $z-1-\psi$ sets consisting only of red agents. Hence, Equation (\ref{eq:red}) becomes
\begin{align}\label{eq:red2}
2r_K+r_X+r_{y_1} = c(2+\psi),
\end{align}
while Equation (\ref{eq:blue}) becomes
\begin{align}\label{eq:blue2}
2b_K+b_X+b_{y_1} = c(z-\psi)-1.
\end{align}

We now argue that $K$ must contain agents of both types. Clearly, since $S\cup J$ hosts $c+1$ blue agents, $K$ (and, by Lemma \ref{cl:perfect-plokami} also $I$) cannot contain only blue agents. If $K$ and $I$ contain only red agents, then there is at least one red agent, let it be agent $i$, in $Y$. The utility of agent $i$ is at most $\frac{b}{b+1}$, as $s$ hosts a blue agent, while by jumping agent $i$ obtains utility at least $\frac{cz/2}{cz/2+c}>\frac{b}{b+1}$, since $z>2b$; hence, $K$ cannot contain only red agents.

Consider a red agent $i$ in $K$; her utility is $\frac{r_K+r_X}{cz/2+c-1}$, while by jumping $i$ would obtain utility $\frac{r_K+r_{y_1}-1}{cz/2+c-1}$. Therefore, we obtain that $r_X\geq r_{y_1}-1$. Similarly, since no blue agent in $K$ wishes to jump, we obtain that $b_X\geq b_{y_1}-1$. Since $r_X+b_X = c-1$ and $r_{y_1}+b_{y_1}=c$, we also have that $r_X\leq r_{y_1}$ and $b_X\leq b_{y_1}$, and, in particular, either $r_X=r_{y_1}$ and $b_X=b_{y_1}-1$ or $r_X=r_{y_1}-1$ and $b_X=b_{y_1}$.

We argue that it must be the case $r_X=r_{y_1}$ and $b_X=b_{y_1}-1$. Indeed, since in $S\cup J \cup \{Y\setminus Y_1\}$ we have already allocated an odd number, i.e., $c\psi+c+1$, of blue agents, and $I\cup K$ must contain an even number of blue agents, an odd number remains to be allocated to $X\cup Y_1$. Equation \ref{eq:red2} becomes
\begin{align}\label{eq:red3}
r_K+r_{y_1} = \frac{c(2+\psi)}{2},
\end{align}
while Equation (\ref{eq:blue2}) becomes
\begin{align}\label{eq:blue3}
b_K+b_{y_1} = \frac{c(z-\psi)}{2}.
\end{align}

Since $b_X=b_{y_1}-1$, we have that $b_{y_1}>0$, while we argue that $r_{y_1}>0$, as otherwise, since $r_X=r_{y_1}$, any red agent in $Y$ has utility $0$ and would prefer to jump; such a red agent exists in $Y$ as, since $S\cup J$ hosts $c+1$ blue agents, it cannot be that all $c(z+1)-1$ agents in $X\cup Y$ are also blue.

Since both $r_{y_1}>0$ and $b_{y_1}>0$, we consider why agents in $Y_1$ do not wish to jump. Since no red agent in $Y_1$ benefits by jumping, we obtain
\begin{align}\label{eq:red_jump}
\frac{r_{x_1}}{b}\geq \frac{r_K+r_{y_1}-1}{cz/2+c-1}
\end{align}
and, since also no blue agent in $Y_1$ wishes to jump,
\begin{align}\label{eq:blue_jump}
\frac{b_{x_1}+1}{b}\geq \frac{b_K+b_{y_1}-1}{cz/2+c-1}.
\end{align}

By combining Equations (\ref{eq:red3})-(\ref{eq:blue_jump}) and, by substituting $c$ and $z$, we obtain
\begin{align}\label{eq:red-final}\nonumber
r_{x_1}&\geq \frac{2(\psi+2)b^3+(\psi+2)b^2-2b}{4b^3+8b^2+3b-2}\\
&= 1+\psi/2 - \frac{(6+3\psi)b^2+(5+3\psi/2)b-\psi-2}{4b^3+8b^2+3b-2},
\end{align}
and
\begin{align}\label{eq:blue-final}\nonumber
b_{x_1}+1&\geq \frac{4b^4+(4-2\psi)b^3+(1-\psi)b^2-2b}{4b^3+8b^2+3b-2}\\
&= b-1-\psi/2 + \frac{(6+3\psi)b^2+(3+3\psi/2)b-\psi-2}{4b^3+8b^2+3b-2}.
\end{align}

Recall that, since the empty node is in $X_1$, it must be $r_{x_1}+b_{x_1} = b-1$. Furthermore, observe that the right-hand-side terms in the two inequalities above are not integers, for any value of $\psi \in \{0, \dots, 2b\}$. The claim follows by observing that, since $r_{x_1}$ and $b_{x_1}$ are integers and $b$ is not a multiple of $3$, it holds that 
\begin{align*}
\lceil \psi/2-\frac{(6+3\psi)b^2+(5+3\psi/2)b-\psi-2}{4b^3+8b^2+3b-2}\rceil +\lceil-\psi/2+\frac{(6+3\psi)b^2+(3+3\psi/2)b-\psi-2}{4b^3+8b^2+3b-2}\rceil =1
\end{align*} for any $\psi \in \{0, \dots, 2b\}$\footnote{Note that this does not hold when $b$ is a multiple of $3$ and $\psi$ is either $2b/3-1$ or $4b/3$.}. Hence, we obtain $r_{x_1}+b_{x_1} = b$; a contradiction.
\end{proof}

By the claims above, we conclude that the empty node is in $Y$; without loss of generality, let the empty node be in $Y_1$. We now show that any equilibrium assignment $\vv'$ is identical (up to symmetry) to $\vv$. Recall that $r_K+b_K = cz/2$, $r_X+b_X = c$, and $r_{x_1}+b_{x_1} = b$.

We first show that all agents in $K$ are red, by exploiting that agents in $K$ have no incentive to jump to the empty node. Let $i\in K$ be a red agent (if one exists in $K$); her utility is $u_i(\vv') = \frac{r_K+r_X}{cz/2+c}$. Since $i$ has no incentive to jump, we obtain
$\frac{r_K+r_X}{cz/2+c} \geq \frac{r_{x_1}}{b+1}$, and, therefore, $\frac{b_K+b_X}{cz/2+c} \leq \frac{b_{x_1}+1}{b+1}$.
Similarly, for any blue agent $j \in K$ (again, assuming one exists) we obtain that $\frac{b_K+b_X}{cz/2+c} \geq \frac{b_{x_1}+1}{b+1}$, and, therefore, $\frac{r_K+r_X}{cz/2+c} \leq \frac{r_{x_1}}{b+1}$.
By these four inequalities above, we conclude that  $\frac{r_K+r_X}{cz/2+c} = \frac{r_{x_1}}{b+1}$ and $\frac{b_K+b_X}{cz/2+c} = \frac{b_{x_1}+1}{b+1}$.
Since $\gcd(cz/2+c, b+1) = 1$, this implies that all agents in $K$ are of the same type. Since each type has $c(z+1)$ agents and all $c+1$ agents in $S\cup J$ are blue, and as by Claim \ref{cl:perfect-plokami} all agents in $I$ are of the same type as those in $K$, we obtain that all $cz/2$ agents in $K$ must be red.

We now show that all agents in $X$ are also red. Assume otherwise and consider a blue agent $j\in X$. Her utility is at most $\frac{2}{z+2}$, while by jumping $j$ would obtain utility at least $\frac{1}{b+1}$ as $s$ is blue. Since $z >2b$, $j$ has an incentive to jump and, therefore, no equilibrium has blue agents in $X$. We have established that all agents in $I\cup K\cup X$ are red and, since there are $c(z+1)$ agents of each type, all remaining agents are blue.

So far, we have completed the argument that $\vv$ is the unique Nash equilibrium, up to symmetry. Under the tolerance vector $\tol{2}' = (1,t_1<1)$, its social welfare is
\begin{align*}
\SW(\vv) & = cz +b(z-1)\frac{cz/2+t_1c}{cz/2+c}+b\frac{cz/2+t_1(c-1)}{cz/2+c-1} +(cz-1)\frac{1+t_1b}{b+1}+c+1\\
&< cz\frac{b(1+t_1)+2}{b+1}+2c+1.
\end{align*}

Consider the following assignment $\vv^*$. All nodes in $I \cup K\cup J$ host red agents, all nodes in $Y \cup S$ host blue agents, while the remaining $c-1$ blue agents are in nodes in $X$. Its social welfare is

\begin{align*}
\SW(\vv^*) &= cz/2(1+\frac{cz/2+t_1(c-1)}{cz/2+c-1})+(c-1)\frac{c+t_1cz/2}{cz/2+c} +cz+\frac{cz}{cz+c}\\
&> cz\frac{z+1+t_1}{z+2}+cz\\
&= cz\frac{2z+3+t_1}{z+2},
\end{align*}
where, in the first equality, the first term is the aggregate utility of agents in $I\cup K$, the second term is the aggregate utility of agents in $X$, the third term is due to agents in $Y$, and the last term is due to agent $s$. Clearly, for the optimal assignment it holds that $\OPT \geq \SW(\vv^*)$.

Since we have selected  $b$ to be arbitrarily large and $z=2b+1$, it holds
\begin{align*}
\frac{\OPT}{\SW(\vv)} &\geq \frac{cz\frac{2z+3+t_1}{z+2}}{cz\frac{b(1+t_1)+2}{b+1}+2c+1}\\
&\geq \frac{2}{1+t_1}-\epsilon,
\end{align*}
for $\epsilon>0$, and the theorem follows.
\end{proof}

\section{Conclusion and future directions}\label{sec:conclusions}
In this paper, we introduced and studied the class of tolerance Schelling games. We made significant progress towards identifying classes of games (with structured topologies such as grids and trees, and binary tolerance vectors)  for which equilibria are guaranteed to exist, and also showed asymptotically tight bounds on the price of anarchy and price of stability. However, there are still many interesting open questions. 

The most important question that our work leaves open is the characterization of games for which equilibria always exist. As this is a quite general and challenging direction, one could start with games that exhibit some sort of structure in terms of the topology or the tolerance vector. For instance, do equilibria always exist when the topology is a grid or a regular graph, for {\em every} tolerance vector? We believe that this is indeed the case, and have done some initial progress towards answering this question in this paper for grids and binary tolerance vectors, but showing a general statement seems elusive at this point. 

The tolerance model we defined in this paper depends on a given ordering of the types and the tolerance parameters are symmetric. While this model captures certain interesting settings, there are multiple ways in which it can be generalized. For example, the tolerance parameters do not need to be symmetric and a different tolerance vector could be defined per type. Taking this further, the tolerance between types does not need to depend on an ordering of the types. Instead, one could define a weighted, directed {\em tolerance graph} that is defined over the different types such that the edge weights indicate the tolerance of a type towards another type; our ordered model can be thought of as the special case with an undirected tolerance line graph. In fact, one could further generalize this idea by considering scenarios in which there are no types of agents at all, but rather the agents are connected to each other via a complete weighted social network, with the different weights indicating tolerance levels.
This is essentially a generalization of the class of social Schelling games proposed by \citet{elkind2019jump}, and is inspired by fractional hedonic games~\citep{aziz2019fractional}.

\bibliographystyle{plainnat}
\bibliography{references}

\end{document}